\documentclass{article}
\usepackage[utf8]{inputenc}

\title{Optimal Variance Reduction in Randomized Experiments}
\author{Amir Najmi \and Michael D. Keselman}
\date{August 21, 2026}

\usepackage{natbib}
\usepackage{graphicx}
\usepackage{amsmath}
\usepackage{amssymb}
\usepackage{amsthm}
\usepackage{hyperref}

\hypersetup{pdfauthor={Name}}

\def\E{{\mathbb{E} \,}}
\def\Var{{\mathbb{V}\mathrm{ar} \,}}
\def\VarDM{\texorpdfstring{\mathrm{\mathbb{V}ar_{DM}} \,}{VarDM}}

\newcommand{\seq}[1]{\{{#1}\}}
\newcommand{\norm}[1]{\widetilde{#1}}
\newcommand{\mean}[1]{\overline{{#1}}}
\newcommand{\metric}{\mathcal{M}}
\newcommand{\metricadj}{\mathcal{M}_{\mathrm{adj}}}
\newcommand{\tmetricadj}{\mathcal{M}_{\mathrm{adj}}^t}

\newcommand{\real}{\mathbb{R}}
\newcommand{\diff}{f}
\newcommand{\diffprimea}{\diff'_a}
\newcommand{\diffprimeb}{\diff'_b}

\newcommand{\diffadj}{h}
\newcommand{\diffadjprimea}{\diffadj'_1}
\newcommand{\diffadjprimeb}{\diffadj'_2}
\newcommand{\diffadjprimec}{\diffadj'_3}
\newcommand{\diffadjprimed}{\diffadj'_4}
\newcommand{\twocolvec}[2]{\begin{bmatrix}{#1} \\ {#2}\end{bmatrix}}
\newcommand{\fourcolvec}[4]{\begin{bmatrix}{#1} \\ {#2} \\ {#3} \\ {#4}\end{bmatrix}}
\newcommand{\tworowvec}[2]{\begin{bmatrix}{#1} & {#2}\end{bmatrix}}
\newcommand{\fourrowvec}[4]{\begin{bmatrix}{#1} & {#2} & {#3} & {#4}\end{bmatrix}}
\begin{document}

\maketitle

\begin{abstract}
    This paper describes an approach to variance reduction in randomized experiments using side information (covariates of the response unaffected by treatment). As with Double ML \citep{doubleml}, models of arbitrary complexity may be employed without concern for bias due to overfitting or regularization.
    For additive treatment metrics, the approach minimizes variance optimally. For more complex treatment metrics (e.g., multiplicative and ratios) the Delta Method estimate of variance is minimized. Through a variety of examples, the paper also explores modeling considerations.
\end{abstract}

\section{Framing}
Consider the tuple of iid observations most generally given by
\[
\left(X, Y(0), Y(1) \right)
\]
Here $X$ is a vector of covariates. $Y(0)$ and $Y(1)$ are responses to control and treatment in the potential outcomes \citep{imbens_rubin} sense. In other words, only one of $Y(0)$ and $Y(1)$ is observed depending on whether a given observation is assigned to control or treatment. While $Y(0)$ and $Y(1)$ are often scalar, it is sometimes useful to extend $Y(0)$ and $Y(1)$ to be a vector.
For example, $Y$ may represent the number of recommendations accepted by a user and the number of recommendations offered to the user. Experiments might seek to measure the overall success rate of recommendations.
This paper deals with the general case, but we start with the simple case of scalar $Y(0)$ and $Y(1)$.

Units of experimentation (e.g., users) are randomly sampled from a population and independently assigned the set $C$ for control and $T$ for treatment. 
The size of these sets is denoted by $|C|$ and $|T|$ and define the fraction of experimental units in treatment to be $\gamma=\frac{|T|}{|T|+|C|}$.
We use indices to denote these sampled observations assigned to sets $C$ and $T$. Note that $Y_i(1)$ is only observed for $i \in T$, and $Y_j(0)$ observed for $j \in C$. While $Y_i(1)$ indicates a random variable, $y_i(1)$ is the (fixed) observation itself. Thus $y_i(0)$ is not available for $i \in T$, nor are $y_j(1)$ for $j \in C$.

\section{Additive treatment effects}
We start with the simpler case of measuring an additive treatment effect. Here, the estimand is
\begin{align*}
\delta = \E Y(1) - \E Y(0)
\end{align*}
The most basic estimator of this effect is the difference in sample means:
\begin{align}
    \metric &= \frac{1}{|T|}\sum_{i \in T}Y_i(1) - \frac{1}{|C|} \sum_{j \in C}Y_j(0)
\end{align}
Our approach to variance reduction involves selecting a function of the covariates $g(x)$ to form a new estimator. Crucially, \textbf{$g$ is chosen without access to the experiment data}.

For additive treatment effects we define a new estimator $\metricadj$:
\begin{align}
    \metricadj &= \metric - \left( 
                \frac{1}{|T|}\sum_{i \in T}g(X_i) - 
                \frac{1}{|C|} \sum_{j \in C}g(X_j) \right)
\end{align}
In expectation, $\metric$ and $\metricadj$ are equal:
\begin{align*}
    \E\metricadj &= \E\metric - \left( 
                \frac{1}{|T|}\sum_{i \in T}\E g(X_i) - 
                \frac{1}{|C|} \sum_{j \in C}\E g(X_j) \right) \\
                &= \E\metric - (\E g(X) - \E g(X)) \\
                &= \E\metric
\end{align*}
Intuitively, the additional term in $\metricadj$ is simply the difference in sample means of $g(X)$ in each arm of the randomized experiment. Since $g$ is a fixed function without knowledge of the experiment data, these sample averages are equal in expectation. Note also that the choice of $g$ is equivalent up to an added constant, since it cancels out in $\metricadj$.

We can now compare the variance of $\metric$ and $\metricadj$:
\begin{align*}
    \Var \metric &= 
        \frac{1}{|T|^2} \sum_{i \in T} \Var Y_i(1)  +
        \frac{1}{|C|^2} \sum_{j \in C} \Var  Y_j(0)
        \tag{independent samples}\\
        &= 
        \frac{1}{|T|^2} (|T| . \Var Y(1) )  +
        \frac{1}{|C|^2} (|C| . \Var Y(0) )
        \tag{iid observations} \\
        &= 
        \frac{1}{|T|} \Var Y(1)  +
        \frac{1}{|C|} \Var Y(0)
\end{align*}
Likewise
\begin{align}
    \Var \metricadj &= 
        \frac{1}{|T|} \Var (Y(1) - g(X))  +
        \frac{1}{|C|} \Var (Y(0) - g(X))\label{eq:var_m_one}
\end{align}
To proceed further, we need the following result.
\newtheorem{lemma}{Lemma}
\begin{lemma}
\label{lemma:optimize}
Let $X \in \real^p$ be a random vector and $Y_0,Y_1 \in \real$ be random scalar variables jointly distributed with $X$. Define $J(f)$ to be a scalar functional as follows
\begin{align*}
    J(f) = \frac{1}{w_1} \Var\left(Y_1 - f(X)\right) + \frac{1}{w_0} \Var\left(Y_0 - f(X)\right) 
\end{align*}
Then $J(f)$ has minimum at
\begin{align*}
    f^*(x) &= (1-\eta) \E[Y_1 | X = x] + \eta \E[Y_0 | X=x]
\end{align*}
where $\eta = w_1/(w_1 + w_0)$. Moreover
\begin{align*}
    J(f) = J(f^*) + \left(\frac{1}{w_1} + \frac{1}{w_0} \right)
    \Var \left(f(X) - f^*(X) \right)
\end{align*}
\end{lemma}
\begin{proof}
\begin{align}
J(f) 
    &= \frac{1}{w_1} \Var\left(Y_1 - f(X)\right) + \frac{1}{w_0} \Var\left(Y_0 - f(X)\right) \notag \\
    &=  \frac{1}{w_1} \Var\left(\E[Y_1|X] - f(X)\right) + 
        \frac{1}{w_0} \Var\left(\E[Y_0|X] - f(X)\right) \notag \\
    &\quad +
       \frac{1}{w_1} \E\left(\Var[Y_1|X]\right) + 
       \frac{1}{w_0} \E\left(\Var[Y_0|X]\right)
    \tag{by Law of Total Variance}
\end{align}
Thus
\begin{align*}
J(f) &= K_1 \left(
        (1-\eta) \Var\left(\E[Y_1|X] - f(X)\right) + 
        \eta \Var\left(\E[Y_0|X] - f(X)\right)
       \right) +  K_2 \\
    K_1 &= \frac{w_1+w_0}{w_0 w_1} = \frac{1}{w_1} + \frac{1}{w_0} \\
    K_2 &=  \frac{1}{w_1} \E\left(\Var[Y_1|X]\right) + 
       \frac{1}{w_0} \E\left(\Var[Y_0|X]\right)
\end{align*}
and $K_1$, $K_2$ do not depend on $f$.
Define the following ``mean-removed'' random variables and function:
\begin{align*}
    \norm{f}(X) &= f(X) - \E f(X) \\
    \norm{Y}_1 &= Y_1 - \E Y_1 \\
    \norm{Y}_0 &= Y_0 - \E Y_0
\end{align*}
The mean-removed version allow us to switch between variance and squared expectation. Hence
\begin{align*}
    J(f)
    &= K_1 \left(
        (1-\eta)\E\left( \E [\norm{Y}_1|X] - \norm{f}(X)\right)^2 + 
        \eta \E\left( \E [\norm{Y}_0|X] - \norm{f}(X)\right)^2
    \right) + K_2 \\
    &= K_1 \E \left[
        (1-\eta)\left( \E [\norm{Y}_1|X] - \norm{f}(X)\right)^2 + 
        \eta \left( \E [\norm{Y}_0|X] - \norm{f}(X)\right)^2
    \right] + K_2
\end{align*}
Consider the binary random variable $U|X$ which has the following distribution:
\begin{align*}
U|X = 
\begin{cases} 
\E [\norm{Y}_1|X]  & \text{with probability } 1-\eta \\
\E [\norm{Y}_0|X] & \text{with probability } \eta 
\end{cases}
\end{align*}
We can write $J(f)$ in terms of $U$ as follows:
\begin{align*}
J(f)
    &= K_1 \E \left[
        \E \left[ \left(U - \norm{f}(X)\right)^2 \bigg| X \right]
    \right] + K_2
\end{align*}
Focusing on the inner expectation
\begin{align*}
\E \left[ \left(U - \norm{f}(X)\right)^2 \bigg| X \right]
&= \Var(U | X) + \left(\norm{f}(X) - \E [U | X] \right)^2 \\
&= \eta (1 - \eta) \left(\E [\norm{Y}_1|X] - \E [\norm{Y}_0|X] \right)^2
+ \left(\norm{f}(X) - \E [U | X] \right)^2
\end{align*}
Now define $\norm{f}^*$
\begin{align*}
    \norm{f}^*(x)  &= \E [U | X=x] \\
    &= (1-\eta) \E[\norm{Y}_1|X=x] + \eta \E[\norm{Y}_0|X=x]
\end{align*}
Adding back in the outer expectation, we have
\begin{align*}
    \E \left[
        \E \left[ \left(U - \norm{f}(X)\right)^2 \bigg| X \right]
    \right]
&= \eta (1 - \eta) \E \left(\E [\norm{Y}_1|X] - \E [\norm{Y}_0|X] \right)^2
+ \E \left(\norm{f}(X) - \norm{f}^*(X) \right)^2 \\
&= \eta (1 - \eta) \Var \left(\E [Y_1|X] - \E [Y_0|X] \right)
+ \Var \left(f(X) - f^*(X) \right)
\end{align*}
where the statement of the lemma defines
\begin{align*}
    f^*(x) &= (1-\eta) \E[Y_1|X=x] + \eta \E[Y_0|X=x]
\end{align*}
Hence
\begin{align*}
J(f) &= K_1 \eta (1 - \eta) \Var \left(\E [Y_1-Y_0|X]\right) + K_1 \Var \left(f(X) - f^*(X) \right) + K_2 \\
&= \frac{1}{w_1+w_0} \Var \left(\E [Y_1-Y_0|X]\right) + K_2 + 
\left(\frac{1}{w_1} + \frac{1}{w_0}\right) \Var \left(f(X) - f^*(X)\right) \\
&= J(f^*) + 
\left(\frac{1}{w_1} + \frac{1}{w_0}\right) \Var \left(f(X) - f^*(X)\right)
\end{align*}
Since variance is non-negative, it follows that $f^*$ minimizes $J$.
\end{proof}

We now simply apply Lemma \eqref{lemma:optimize} to Equation \eqref{eq:var_m_one} to find that $\Var \metricadj$ is minimized over $g$ when
\begin{align}
    g(x)  &= (1-\gamma) \E[Y(1)|X=x] + \gamma \E[Y(0)|X=x] \label{eq:g_diff}
\end{align}

It is worth pondering this expression:
\begin{itemize}
    \item The best choice of $g$ is a convex combination of the conditional expectations of treated and untreated responses $Y(1)$, $Y(0)$. These represent the  portion of the response that can be predicted by the covariate $X$.
    
    \item We have a single $g$ for both treatment and control, so it needs to be a compromise between treatment and control.
    
    \item The weights in this convex combination are ``backwards'', i.e., weight $\gamma$ for control and $(1-\gamma)$ for treatment. This makes sense if we imagine a very small treatment arm $\gamma \to 0$. In this case, the uncertainty in $\metricadj$ comes almost entirely from the treatment arm, and hence $g$ needs to focus on the treatment at the expense of control (vice versa for $\gamma \to 1$).
    \end{itemize}

\subsection{Algorithm: additive treatment effect}
\label{section:additive_algorthm}
From Equation \eqref{eq:g_diff} we know how to compute $g$, but it requires knowledge of $Y(W)|X$ where $W \in (0, 1)$ represents the assignment, i.e., either control $W=0$ or treatment $W=1$. We do this by modeling on held-out data. This is easiest achieved using $K$-fold cross-fitting (described below) resulting in $K$ estimates. These estimates are not independent, but each is unbiased. As such, they can be averaged to produce a single unbiased estimate.
\begin{enumerate}
    \item Divide the data randomly into $K$ equal folds. Let $(T_k, C_k)$ represent the samples in treatment and control for the $k^\mathrm{th}$ fold.
    
    \item Set the training data to all data except the $k^\mathrm{th}$ fold, i.e., $(T-T_k, C-C_k)$.
    
    \item Without getting into modeling considerations, fit a model of the form $\Hat{y}(W,x)$ that minimizes $(\Hat{y}-y)^2$ on the training data. There are no constraints on the structure of this model in terms of how it trains, regularizes or borrows strength.
    
    \item Using this model, compute
    \begin{align}
    g(x) &= (1-\gamma)\Hat{y}(W=1,x) + \gamma \Hat{y}(W=0,x)\label{eq:gnorm}
    \end{align}
    for each observation in $T_k \cup C_k$. We do not necessarily observe $y(0)$ and $y(1)$ for the same $x$, which is why we do not fit a model for $g(x)$ directly.
    
    \item Estimate
    \begin{align*}
        \metricadj^k &=  \left(\frac{1}{|T_k|}\sum_{i \in T_k}y_i - 
                               \frac{1}{|C_k|}\sum_{j \in C_k}y_j \right) \\
                    &\quad -
                         \left(\frac{1}{|T_k|}\sum_{i \in T_k}g(x_i) - 
                               \frac{1}{|C_k|}\sum_{j \in C_k}g(x_j) \right)
    \end{align*}

    \item Repeat for each fold and form the average $\mean{\metricadj}$. This is the final estimate.
\end{enumerate}
Note that each estimate $\metricadj^k$ is unbiased, even though estimates are not independent. As a result, their 
average $\mean{\metricadj}$ is also unbiased.
Where the data is big (but effects small), $K=2$ will suffice.

\section{Asymptotic variance}
The variance of more complex treatment effects may not have simple closed forms, making them harder to optimize. The approach taken in this paper is to optimize the asymptotic variance of a $\sqrt{n}$-consistent estimator. This approach is appropriate for the large samples we find in experiments run by large scale online services. For small sample experiments, it may serve as a heuristic when applicable.

Thus, we introduce the notion of $\VarDM$ that quantifies asymptotic variance.

\theoremstyle{definition}
\newtheorem*{definition}{Definition}

\subsection{Delta Method Variance $\VarDM$}
\theoremstyle{definition}
$\VarDM$ is an operator on sequences of random variables.
Denote the sequence of random variables $A_1, A_2, \cdots, A_n$ by $\seq{A}$. 
We further denote the following:
\begin{align*}
\seq{c^T A} &:= c^T A_1, c^T A_2, \cdots, c^T A_n \\
\seq{h(A)} &:= h(A_1), h(A_2), \cdots, h(A_n) \\
\seq{A+B} &:= (A_1 + B_1), (A_2 + B_2), \cdots, (A_n + B_n) 
\end{align*}
for any constant $c$ and any function $h$.

\begin{definition}
Let $\seq{A}$ be a sequence of random variables such that
\begin{align*}
    A_n &\xrightarrow{\mathcal{P}} \theta_A \\
    \sqrt{n}(A_n -\theta_A) &\xrightarrow{\mathcal{D}} X_A
\end{align*}
for some constant $\theta_A$ and some random variable $X_A$ for which $\E X_A = 0$ and $\Var{X_A}$ exists, then by definition
\begin{align*}
    \VarDM \seq{A} := \Var{X_A}
\end{align*}
Otherwise $\VarDM\seq{A}$ is undefined.
\end{definition}

\subsection{Properties of $\VarDM$}
Next are a few properties of $\VarDM$ that we will need in later sections. The most important of these is based on the Delta Method theorem.
\newtheorem{theorem}{Theorem}
\begin{theorem}
\label{theorem:delta_method}
Let $\seq{A}$ be a sequence of random variables such that $A_n \xrightarrow{\mathcal{P}} \theta_A$ and
$\VarDM \seq{A}$ exists. Let $h$ be a differentiable function such that $h'(\theta_A) \neq 0$ and is defined. Then
\begin{align*}
    \VarDM \seq{h(A)}  &= \VarDM \seq{h'(\theta_A)^T A}
\end{align*}
\end{theorem}
\begin{proof}
Since $\VarDM \seq{A}$ exists, define $\VarDM \seq{A} = \Sigma$.
From the definition of $\VarDM$, there exists a random variable $X_A$ such that
\begin{align*}
    \sqrt{n}(A_n -\theta_A) &\xrightarrow{\mathcal{D}} X_A \\
    \E X_A &= 0 \\
    \Var X_A &= \Sigma
\end{align*}
Consider the sequence $\seq{c^T A}$ where $c$ is a constant. This sequence must converge in distribution to $ c^T X_A$ (by Continuous Mapping Theorem).
Thus
\begin{align}
    \VarDM \seq{c^T A}  &= \Var( c^T X_A ) \notag \\
    &= c^T \Sigma c \label{eq:varDM_cA}
\end{align}
Now consider the sequence $\seq{h(A)}$.
According to the Delta Method Theorem, there is a random variable $X_h$ such that:
\begin{align*}
    \sqrt{n}(h(A_n) - h(\theta_A)) &\xrightarrow{\mathcal{D}} X_h \\
    \E X_h &= 0 \\
    \Var X_h &= h'(\theta_A)^T \Sigma h'(\theta_A) \\
    \therefore \VarDM \seq{h(A)} &= h'(\theta_A)^T \Sigma h'(\theta_A)
\end{align*}
Further note that when $c = h'(\theta_A)$ in \eqref{eq:varDM_cA}, we have
\begin{align*}
    \VarDM \seq{h'(\theta_A)^T A} &= h'(\theta_A)^T \Sigma h'(\theta_A) \\
    \implies \VarDM \seq{h(A)} &= \VarDM \seq{h'(\theta_A)^T A}
\end{align*}
\end{proof}

\begin{lemma}
\label{lemma:vardm_sum}
Let $\seq{A}$ and $\seq{B}$ be sequences of random variables such that $A_n$ and $B_n$ are independent.
If $\VarDM \seq{A}$ and $\VarDM \seq{B}$ both exist then for any scalar constants $a$ and $b$
\begin{align*}
    \VarDM \seq{a A + b B} &= a^2 \VarDM \seq{A} + b^2 \VarDM \seq{B}
\end{align*}
\end{lemma}
\begin{proof}
Since $\VarDM \seq{A}$ and $\VarDM \seq{B}$ both exist, it follows that
\begin{align*}
    A_n &\xrightarrow{\mathcal{D}} X_A \\
    B_n &\xrightarrow{\mathcal{D}} X_B \\
    \therefore \seq{a A+ b B} &\xrightarrow{\mathcal{D}} a X_A + b X_B \tag{by Continuous Mapping Theorem} \\
    \VarDM\seq{A} &= \Var X_A \tag{by definition of $\VarDM\seq{A}$}\\
    \VarDM\seq{B} &= \Var X_B  \tag{by definition of $\VarDM\seq{B}$}\\
    \VarDM\seq{a A + b B} &= \Var (a X_A + b X_B) \tag{by definition of $\VarDM\seq{a A + b B}$} \\
    &= \Var (a X_A) + \Var (b X_B) \tag{by independence} \\
    &= a^2 \Var X_A + b^2 \Var X_B \\
    &= a^2 \VarDM\seq{A} + b^2 \VarDM\seq{B}
\end{align*}
\end{proof}

\subsection{$\VarDM$ for sequences of sample means}
The discussion of $\VarDM$ thus far has assumed $\sqrt{n}$-convergence for sequences without mentioning the origin of such sequences. Of course, such sequences arise naturally when the Central Limit Theorem applies to the mean of independent samples.

\newtheorem{corollary}{Corollary}[theorem]
\begin{theorem}
\label{theorem:xbar}
Let $X$ be a random variable obeying the Central Limit Theorem.
Let $\lambda$ be a positive scalar.
If $\mean{X}_n$ is the sample mean of $\lceil n \lambda \rceil$ independent realizations of $X$ then
\begin{align*}
    \VarDM \seq{\mean{X}} &= \frac{1}{\lambda}\Var X
\end{align*}
\end{theorem}
\begin{corollary}
\label{corollary:xbar}
If $\mean{X}_n$ is the sample mean of $n$ independent realizations of $X$ then
\begin{align*}
    \VarDM \seq{\mean{X}} &= \Var X
\end{align*}
\end{corollary}
\begin{proof}
We consider the general case, since the corollary follows from the theorem when $\lambda=1$.
Let $\E X = \mu$ and $\Var X = \Sigma$.
If $\mean{X}_n$ is the sample mean of $\lceil n \lambda \rceil$ independent realizations of $X$ then from the Central Limit Theorem
\begin{align}
    \sqrt{\lceil n \lambda \rceil}(\mean{X} - \mu) &\xrightarrow{\mathcal{D}} U \label{eq:norm_lambda_ceiling}\\
\textrm{where\ } U &\sim \mathcal{N}(0, \Sigma) \notag
\end{align}
Furthermore
\begin{align}
    \frac{n}{\lceil n \lambda \rceil} &= \frac{n}{n \lambda + \epsilon}
        \tag{where $0 \leq \epsilon < 1$}\\
        &= \frac{1}{\lambda + \epsilon/n} \notag \\
    \therefore \lim_{n \to \infty} \sqrt{\frac{n}{\lceil n \lambda \rceil}} 
        &= \frac{1}{\sqrt{\lambda}} \label{eq:limit_lambda_ceiling}\\
    \sqrt{n}(\mean{X} - \mu) &=
        \left( \sqrt{\frac{n}{\lceil n \lambda \rceil}}\right)
        \left( \sqrt{\lceil n \lambda \rceil}(\mean{X} - \mu) \right) \notag \\
    \implies \sqrt{n}(\mean{X} - \mu)
        &\xrightarrow{\mathcal{D}}  \frac{1}{\sqrt{\lambda}} U
        \tag{by \eqref{eq:limit_lambda_ceiling}, \eqref{eq:norm_lambda_ceiling} and
        Slutsky's Theorem}
\end{align}
$\VarDM \seq{\mean{X}}$ is the variance of the limiting distribution of $\sqrt{n}(\mean{X} - \mu)$.
Hence
\begin{align*}
    \VarDM \seq{\mean{X}} &= \Var \left( \frac{1}{\sqrt{\lambda}} U \right) \\
        &= \frac{1}{\lambda}\Sigma \\
        &= \frac{1}{\lambda}\Var X
\end{align*}
\end{proof}
These are all the results involving $\VarDM$ we need to pursue asymptotically optimal variance reduction.

\section{Generalization}
So far, we have demonstrated an optimal (asymptotic) variance reduction strategy for additive treatment effects. We chose what might seem to be a natural adjustment to the statistic measured for each arm of the experiment, namely, $Y - g(X)$. But the range of possibilities for both treatment effect and adjustment is vast. We would like to know that we have the optimal adjustment in each scenario. For this purpose, we need a more elaborate setup.

As before, observations are iid tuples of the form
\[
\left(X, Y(0), Y(1) \right)
\]
but now $Y \in \real^p$. Let $\mu_1 = \E Y(1)$ and $\mu_0 = \E Y(0)$.
Let $\diff:\real^p \times \real^p \to \real$ be a differentiable function that measures the treatment effect relative to control, i.e., the treatment effect. The estimand is defined as
\begin{align}
    \delta &= \diff(\mu_1, \mu_0)
\end{align}

The estimand $\delta$ is estimated using sample means $\mean{Y}_1$ and $\mean{Y}_0$:
\begin{align}
    \metric &= \diff(\mean{Y}_1, \mean{Y}_0)
\end{align}
Under the assumptions for $\diff$, $\metric$ is a consistent estimator of $\delta$:
\begin{align*}
    \metric &= \diff(\mean{Y}_1, \mean{Y}_0)\\
    & \xrightarrow{\mathcal{P}} \diff(\mu_1, \mu_0) \tag{by Continuous Mapping Theorem}\\
    &= \delta
\end{align*}
Variance reduction takes the form of an adjustment to $\diff(\mean{Y}_1, \mean{Y}_0)$ using a function $g(X)$ which is a vector, i.e., $g(X) \in \real^q$. Let $\mean{G}_1$, $\mean{G}_0$ be the average of $g(X)$ over the treatment sample and control sample respectively. Given that $X$ does not depend on assignment to treatment or control, $\E \mean{G}_1 = \E \mean{G}_0 = \mu_g$.

\hfill
\\
Let the adjusted metric $\metricadj$ be given by the function $\diffadj : \real^p \times \real^p \times \real^q \times \real^q \to \real$:
\begin{align*}
    \metricadj &= \diffadj(\mean{Y}_1, \mean{Y}_0, \mean{G}_1, \mean{G}_0)
\end{align*}
A required property of $\diffadj$ is that a common adjustment ``cancels out'':
\begin{align}
    \diffadj(a, b, c, c) &= \diff(a, b)\label{eq:cancel}
\end{align}
This \textbf{cancellation property} guarantees consistency of $\metricadj$.
Our approach is to define the size of $T$ and $C$ in terms of an index $n$
such that
 $|T| = \lceil n \lambda_1\rceil$ and $|C| = \lceil n \lambda_0\rceil$ respectively,
where $\lambda_1, \lambda_0 > 0$.
Thus, each of
$\seq{\mean{Y}_1}$, $\seq{\mean{Y}_0}$, $\seq{\mean{G}_1}$, $\seq{\mean{G}_0}$, 
$\seq{\metric}$, $\seq{\metricadj}$ 
is a sequence indexed by $n$.
Assuming all necessary derivatives exist, the cancellation property ensures that $\metricadj$ is a consistent estimator of $\delta$:
\begin{align*}
    \metricadj &= \diffadj(\mean{Y}_1, \mean{Y}_0, \mean{G}_1, \mean{G}_0)\\
    & \xrightarrow{\mathcal{P}} \diffadj(\mu_1, \mu_0, \mu_g, \mu_g) \tag{Continuous Mapping Theorem}\\
    &= \diff(\mu_1, \mu_0) \tag{by Equation \eqref{eq:cancel}}\\
    &= \delta
\end{align*}
\\
Given $\diff$ and $\diffadj$, our objective is to find $g$ to minimize $\VarDM \seq{\metricadj}$.

\begin{theorem}
\label{theorem:General}
Let $Y(1), Y(0) \in \real^p$ be random variables with expectation $\mu_1$, $\mu_0$ respectively.\\
Let $\diff(a,b)$ be a differentiable function with non-zero derivative at $(\mu_1, \mu_0)$. We use the notation $\diffprimea$ and $\diffprimeb$ to denote the partial derivatives of $\diff$ with respect to its first and second argument. $\diffprimea(a,b)$, $\diffprimeb(a,b) \in \real^p$.\\
Let $X$ be a random variable and $g(x) \in \real^q$ be a function defined over its domain such that
$\E g(X) = \mu_g$.\\
Let $\diffadj(a, b, c, d)$ be a differentiable function such that for all $a, b, c$
\begin{align*}
    \diffadj(a, b, c, c) &= \diff(a, b)
\end{align*}
and that $\diffadj$ has non-zero partial derivatives at $(\mu_1, \mu_0, \mu_g, \mu_g)$. We denote partial derivatives of $\diffadj$ with respect to its four arguments as
$\diffadjprimea$, $\diffadjprimeb$, $\diffadjprimec$ and $\diffadjprimed$.
These derivatives are vectors in $\real^p$, $\real^p$, $\real^q$ and $\real^q$ respectively.

\hfill
\\
Let $T$ and $C$ be sets of iid samples of $(X, Y(1))$ and $(X, Y(0))$.\\
Let $\mean{Y}_1$, $\mean{G}_1$ be the sample averages of $Y(1), g(X)$ in $T$ \\
and $\mean{Y}_0$, $\mean{G}_0$ be the sample averages of $Y(0), g(X)$ in $C$.

\hfill
\\
Define
\begin{align}
    \metricadj &= \diffadj(\mean{Y}_1, \mean{Y}_0, \mean{G}_1, \mean{G}_0) \label{eq:def_Madj}
\end{align}
Let the size of $T$, $C$ be $\lceil n \lambda_1\rceil$ and $\lceil n \lambda_0\rceil$ respectively,
where $\lambda_1,\lambda_0 > 0$. \\
$\seq{\mean{Y}_1}$, $\seq{\mean{Y}_0}$, $\seq{\mean{G}_1}$, $\seq{\mean{G}_0}$, $\seq{\metricadj}$ are
sequences induced by $n$.

\hfill
\\
Then $\VarDM\seq{\metricadj}$ exists and is minimized over choice of $g = g^*$
\begin{align*}
v_g ^T g^*(x)  &= (1 - \gamma) v_a ^T \E[Y(1)|X=x] 
                - \gamma v_b ^T \E[Y(0)|X=x]
\end{align*}
where $\gamma = \lambda_1/(\lambda_1 + \lambda_0)$ and
\begin{align}
v_a &= \diffprimea(\mu_1, \mu_0) \label{eq:def_va}\\
v_b &= \diffprimeb(\mu_1, \mu_0) \label{eq:def_vb}\\
v_g &= \diffadjprimed(\mu_1, \mu_0, \mu_g, \mu_g) \label{eq:def_vg}
\end{align}
Furthermore, the minimum can be achieved for scalar $g$ (i.e., $q=1$).
Finally, the use of scalar $g \neq g^*$ results in additional variance
$\Var(v_g (g(X) - g^*(X))$.
\end{theorem}
\begin{proof}
From these definitions and assumptions, we may infer the following:
\begin{align}
\frac{d }{d c} \diffadj(a, b, c, c) &=
    \diffadjprimec(a, b, c, c) + \diffadjprimed(a, b, c, c) \notag\\
\frac{d }{d c} \diffadj(a, b, c, c) &= \frac{d }{d c} \diff(a, b) 
\tag{by Equation \eqref{eq:cancel}} = 0 \notag\\
\implies \diffadjprimec(a, b, c, c) &= -\diffadjprimed(a, b, c, c) \label{eq:antisym}
\end{align}
Also
\begin{align}
\left. \frac{d }{d u} \diffadj(u, b, c, c)) \right|_{u=a}&=
    \diffadjprimea(a, b, c, c) \notag\\
\left. \frac{d }{d u} \diffadj(u, b, c, c)) \right|_{u=a}&=
\left. \frac{d }{d u} \diff(u, b) \right|_{u=a}\notag = \diffprimea(a, b) \notag \\
\implies \diffadjprimea(a, b, c, c) &= \diffprimea(a, b) \label{eq:diffprimea}
\end{align}
Likewise from the derivative of $\diffadj(a, u, c, c)$ w.r.t $u$ at $u=b$
\begin{align}
\diffadjprimeb(a, b, c, c) &= \diffprimeb(a, b) \label{eq:diffprimeb}
\end{align}
Collecting together these results for the derivative of $\diffadj$ we have
\begin{align*}
    \diffadj'(\mu_1, \mu_0, \mu_g, \mu_g) &= 
    \begin{bmatrix}
    \diffadjprimea(\mu_1, \mu_0, \mu_g, \mu_g) \\
    \diffadjprimeb(\mu_1, \mu_0, \mu_g, \mu_g) \\
    \diffadjprimec(\mu_1, \mu_0, \mu_g, \mu_g) \\
    \diffadjprimed(\mu_1, \mu_0, \mu_g, \mu_g)
    \end{bmatrix} \\
    &= 
    \fourcolvec{\diffprimea(a, b)}{\diffprimeb(a, b)}{-v_g}{v_g}
    \tag{by Equations \eqref{eq:def_vg}, \eqref{eq:antisym}, \eqref{eq:diffprimea}, \eqref{eq:diffprimeb}}\\
    &= 
    \fourcolvec{v_a}{v_b}{-v_g}{v_g}
   \tag{by Equations \eqref{eq:def_va}, \eqref{eq:def_vb} }\\
\end{align*}

Since the derivative of $\diffadj$ at $(\mu_1, \mu_0, \mu_g, \mu_g)$ is non-zero by assumption
\begin{align*}
    \VarDM \seq{\metricadj} &= 
    \VarDM \seq{\diffadj(\mean{Y}_1, \mean{Y}_0, \mean{G}_1, \mean{G}_0)}
    \tag{by Equation \eqref{eq:def_Madj}} \\
    &= \VarDM \seq{
        \diffadj'(\mu_1, \mu_0, \mu_g, \mu_g)^T 
        \fourcolvec{\mean{Y}_1}{\mean{Y}_0}{\mean{G}_1}{\mean{G}_0} 
        }
        \tag{by Theorem \eqref{theorem:delta_method}} \\
    &= \VarDM \seq{
                \fourrowvec{v_a}{v_b}{-v_g}{v_g}  
                \fourcolvec{\mean{Y}_1}{\mean{Y}_0}{\mean{G}_1}{\mean{G}_0} 
            } \\
    &= \VarDM \seq{ v_a ^T \mean{Y}_1 + v_b ^T \mean{Y}_0 - v_g ^T \mean{G}_1 + v_g ^T \mean{G}_0 }
\end{align*}
Since $\mean{Y}_1, \mean{G}_1$ (treatment) are independent of $\mean{Y}_0, \mean{G}_0$ (control), it follows that
\begin{align}
\VarDM \seq{\metricadj}
&= \VarDM \seq{ v_a ^T \mean{Y}_1 + v_b ^T \mean{Y}_0 - v_g ^T \mean{G}_1 + v_g ^T \mean{G}_0 }
    \notag \\
    &= \VarDM \seq{ v_a ^T \mean{Y}_1 - v_g ^T \mean{G}_1 }
    + \VarDM \seq{ v_b ^T \mean{Y}_0 + v_g ^T \mean{G}_0 }
        \tag{by Lemma \eqref{lemma:vardm_sum}} \\
    &= \frac{1}{\lambda_1} \Var \left(v_a ^T Y(1) - v_g ^T g(X) \right) + 
       \frac{1}{\lambda_0} \Var\left( v_b ^T Y(0) + v_g ^T g(X) \right)
    \tag{by Theorem \eqref{theorem:xbar}} \\
    &= \frac{1}{\lambda_1} \Var \left(v_a ^T Y(1) - v_g ^T g(X) \right) + 
       \frac{1}{\lambda_0} \Var \left(-v_b ^T Y(0) - v_g ^T g(X) \right) 
    \label{eq:vardm_optim}
\end{align}
Each of $v_g ^T g(x)$, $v_a ^T Y(1)$ and $v_b ^T Y(0)$ are scalar.
Thus, we apply Lemma \eqref{lemma:optimize} with substitutions
$f(x) \rightarrow v_g^T g(x)$, $Y_1 \rightarrow v_a^T Y(1)$, $Y_0 \rightarrow -v_b^T Y(0)$
to assert that $\min_g \VarDM \seq{\metricadj}$ occurs when
\begin{align}
    v_g ^T g(x) &= (1-\gamma) v_a ^T \E [Y(1)|X=x] 
                        - \gamma v_b ^T \E [Y(0)|X=x]
    \label{eq:g_optim}
\end{align}
Furthermore,
$v_a = \diffprimea(\mu_1, \mu_0)$ and $v_b = \diffprimeb(\mu_1, \mu_0)$ do not depend on the form of the adjusted metric $\diffadj$. Only $v_g$ does, being the marginal change in $\diffadj$:
\begin{align*}
    v_g &= \diffadjprimed(\mu_1, \mu_0, \mu_g, \mu_g)\\
    &= \left. \frac{d}{d u} \diffadj(\mu_1, \mu_0, \mu_g, u) ) \right|_{u=\mu_g}
\end{align*}
Equation \eqref{eq:vardm_optim} shows that $\VarDM \seq{\metricadj}$ depends on $g$ only through
$v_g ^T g(x)$.
By assumption, 
$g(x) \in \real^q$, 
but variance is minimum as long as $v_g ^T g(x)$ has the optimal value.
Thus, we may choose $g(x)$ to be scalar leading to scalar $v_g$.
Finally, Lemma \eqref{lemma:optimize} also tells us the penalty for using
$\widehat{g}$ instead of $g$:
\begin{align}
\VarDM \seq{\metricadj} &= \VarDM \seq{\metricadj^*} + 
\left(\frac{1}{\lambda_1} + \frac{1}{\lambda_0} \right)
\Var \left(v_g (\widehat{g}(X) - g(X)) \right) \label{eq:g_suboptim}
\end{align}
\end{proof}

Next is a theorem to extend first order optimal $g$ to related metrics.
\begin{theorem}
\label{theorem:optimal_invariance}
Let the quantities and conditions be as defined in Theorem \eqref{theorem:General}.\\
Let $t$ be any differentiable scalar function such that
$t'(\delta) \neq 0$ for $\delta = f(\mu_1, \mu_0) = \diffadj(\mu_1, \mu_0, \mu_g, \mu_g)$.
Let
\begin{align*}
    \tmetricadj &= t(\metricadj)
\end{align*}
If $g$ minimizes $\VarDM\seq{\metricadj}$ then it also minimizes $\VarDM\seq{\tmetricadj}$.
\end{theorem}
\begin{proof}
\begin{align*}
    \VarDM \seq{\metricadj} &= 
    \VarDM \seq{\diffadj(\mean{Y}_1, \mean{Y}_0, \mean{G}_1, \mean{G}_0)}
    \tag{by Equation \eqref{eq:def_Madj}} \\
    &= \VarDM \seq{
        \diffadj'(\mu_1, \mu_0, \mu_g, \mu_g)^T \mean{B}
        }
        \tag{by Theorem \eqref{theorem:delta_method}}
\end{align*}
where
$\mean{B} = \fourrowvec{\mean{Y}_1}{\mean{Y}_0}{\mean{G}_1}{\mean{G}_0}^T$.
Similarly
\begin{align*}
    \VarDM \seq{\tmetricadj} &= 
    \VarDM \seq{t(\diffadj(\mean{Y}_1, \mean{Y}_0, \mean{G}_1, \mean{G}_0))} \\
    &= \VarDM \seq{
        t'(\diffadj(\mu_1, \mu_0, \mu_g, \mu_g))
        \diffadj'(\mu_1, \mu_0, \mu_g, \mu_g)^T \mean{B}} \\
    &= \VarDM \seq{
        t'(\delta)
        \diffadj'(\mu_1, \mu_0, \mu_g, \mu_g)^T \mean{B}} \\
    &= t'(\delta)^2
        \VarDM \seq{
        \diffadj'(\mu_1, \mu_0, \mu_g, \mu_g)^T \mean{B}} 
        \tag{by Lemma \eqref{lemma:vardm_sum}} \\
    &= t'(\delta)^2 \VarDM \seq{\metricadj}
\end{align*}
\end{proof}

\subsection{Constraints on $g$ and $\diffadj$}
Theorem \ref{theorem:General} provides a condition for optimal $g$. Nevertheless, there are decisions to be made because
\begin{itemize}
    \item there is often considerable choice in how $\diffadj$ is defined to use $\mean{G}_1, \mean{G}_0$ while still respecting the cancellation property. For instance, $\diffadj$ may be an additive function of $\mean{G}_1 - \mean{G}_0$ or of $r(\mean{G}_1) - r(\mean{G}_0)$ or even $s(\mean{G}_1 - \mean{G}_0)$ where $r$ and $s$ are differentiable and $s(0)=0$.
    \item the stated optimality is defined to first order. Thus all functions that have the same first order Taylor series expansion near their mean are equivalent for this purpose.
\end{itemize}
The treatment presented does not suggest a specific choice. Unless there are clear reasons of model fit to choose a particular form, simpler is usually better.

However, there is one constraint that can eliminate certain choices of $\diffadj$ and $g$. Taking expectations on both sides of Equation \eqref{eq:g_optim} leads to
\begin{align}
    v_g ^T \mu_g &= (1-\gamma) v_a ^T \E Y(1) - \gamma v_b ^T \E Y(0)
    \label{eq:constraint_optim_g}
\end{align}
where $v_g$ depends on $\mu_g$ and $\diffadj$ while $v_a$ and $v_b$ do not. For certain choices of $\diffadj$, there may be no choice of $g$ that can satisfy Equation \eqref{eq:constraint_optim_g}.

It is worth noting that there is always a solution if $h$ is defined simply as an additive adjustment to $f$:
\begin{align}
h(\mean{Y}_1, \mean{Y}_0, \mean{G}_1, \mean{G}_0)
&= f(\mean{Y}_1, \mean{Y}_0) - \mean{G}_1 + \mean{G}_0 \label{eq:simple_h}\\
\Rightarrow \metricadj &= \metric - \mean{G}_1 + \mean{G}_0\notag
\end{align}
In this case $v_g = \diffadjprimed(\mu_1, \mu_0, \mu_g, \mu_g) = 1$, and by Theorem \eqref{theorem:General}:

\begin{align}
g(x) &= (1-\gamma) v_a^T \E[Y(1)|X=x] -\gamma v_b^T\E[Y(0)|X=x]  \label{eq:simple_optim}
\end{align}

\section{Applications}
In this section, we apply the general result on asymptotic optimality to specific situations of interest.

\subsection{Additive treatment effects}
\label{section:addtive}
We start with the case of additive treatment effects for which derived optimal non-asymptotic results in Equation \eqref{eq:g_diff}.
In this setting, $Y$ is scalar and the treatment effect is defined by the difference $\E Y(1) - \E Y(0)$. Thus
\begin{align*}
    \delta &= f(\mu_1, \mu_0) = \mu_1 - \mu_0 \\
    \metric &= f(\mean{Y}_1, \mean{Y}_0) = \mean{Y}_1 - \mean{Y}_0 \\
    \metricadj &= h(\mean{Y}_1, \mean{Y}_0, \mean{G}_1, \mean{G}_0)
    = (\mean{Y}_1 - \mean{G}_1) - (\mean{Y}_0 - \mean{G}_0)
\end{align*}
As required $\metricadj = \metric$ when $\mean{G}_1=\mean{G}_0$. Thus
\begin{align*}
v_a &= \diffprimea(\mu_1, \mu_0) = 1 \tag{by Equation \eqref{eq:def_va}}\\
v_b &= \diffprimeb(\mu_1, \mu_0) = -1 \tag{by Equation \eqref{eq:def_vb}}\\
v_g &= \diffadjprimed(\mu_1, \mu_0, \mu_g, \mu_g) = 1 \tag{by Equation \eqref{eq:def_vg}} \\
\therefore g(x) &= (1-\gamma) \E[Y(1)|X=x] + \gamma \E[Y(0)|X=x] \tag{by Theorem \eqref{theorem:General}}
\end{align*}
This is identical to Equation \eqref{eq:g_diff}. In other words, our asymptotically optimal solution is (non-asymptotically) optimal in this simple common case.

\subsection{Multiplicative treatment effect: log scale}
\label{section:mult_simple}
Much of the literature on variance reduction, indeed experiment analysis and causal inference itself, posits additive treatment effects. However,  multiplicative effects are often of greater practical interest in domains such as large scale online services.

In the simplest case, $Y$ is scalar and the effect is defined on a log scale:
\begin{align*}
    \delta &= \log \mu_1 - \log \mu_0
\end{align*}
This naturally leads to a plug-in estimator
\begin{align*}
    f(\mean{Y}_1, \mean{Y}_0) 
    &= \log \mean{Y}_1 - \log \mean{Y}_0 \\
    &= \metric
\end{align*}
where $\mean{Y}_1$, $\mean{Y}_0$ are the sample mean of $Y$ in treatment and control respectively.
Using the vector of covariates $X$ and function $g(X)$ we define an adjusted estimator
\begin{align}
\metricadj 
    &= (\log \mean{Y}_1 - \mean{G}_1) - (\log \mean{Y}_0 - \mean{G}_0)
    \label{eq:def_Madj_mult_bis}
\end{align}
where $\mean{G}_1$, $\mean{G}_0$ are the sample mean of $g(X)$ in treatment and control respectively.
This is an instance of Equation \eqref{eq:simple_h}, where
\begin{align*}
v_a &= \diffprimea(\mu_1, \mu_0) = \frac{1}{\mu_1} \tag{by Equation \eqref{eq:def_va}}\\
v_b &= \diffprimeb(\mu_1, \mu_0) = -\frac{1}{\mu_0} \tag{by Equation \eqref{eq:def_vb}}
\end{align*}
and hence
\begin{align}
    g(x) &= (1-\gamma) \frac{\E[Y(1)|X=x]}{\E Y(1)} +\gamma \frac{\E[Y(0)|X=x]}{\E Y(0)}
    \tag{by Equation \eqref{eq:simple_optim}}
\end{align}
where $\E Y(1)$, $\E Y(0)$ are $\mu_1$, $\mu_0$.

Alternatively, we could use a different adjusted estimator
\begin{align}
    h(\mean{Y}_1, \mean{Y}_0, \mean{G}_1, \mean{G}_0)
    &= (\log \mean{Y}_1 - \log \mean{G}_1) - (\log \mean{Y}_0 - \log \mean{G}_0) \notag \\
    \metricadj 
    &= (\log \mean{Y}_1 - \log \mean{G}_1) - (\log \mean{Y}_0 - \log \mean{G}_0)
    \label{eq:def_Madj_mult}
\end{align}
Notice the log in front of $\mean{G}_1$ and $\mean{G}_0$. As required $\diffadj$ has the cancellation property (Equation \eqref{eq:cancel}) such that $\metricadj = \metric$ when $\mean{G}_1 = \mean{G}_0$.
This time
\begin{align*}
v_a &= \diffprimea(\mu_1, \mu_0) = \frac{1}{\mu_1} \tag{by Equation \eqref{eq:def_va}}\\
v_b &= \diffprimeb(\mu_1, \mu_0) = -\frac{1}{\mu_0} \tag{by Equation \eqref{eq:def_vb}}\\
v_g &= \diffadjprimed(\mu_1, \mu_0, \mu_g, \mu_g) = \frac{1}{\mu_g} \tag{by Equation \eqref{eq:def_vg}} \\
\frac{g(x)}{\mu_g} &= 
        (1-\gamma) \frac{\E[Y(1)|X=x]}{\mu_1} + \gamma \frac{\E[Y(0)|X=x]}{\mu_0}
        \tag{by Theorem \eqref{theorem:General}}
\end{align*}
Since $\mu_g = \E g(X)$, optimal $g$ is defined up to a multiplicative constant --- a fact also evident from
 Equation \eqref{eq:def_Madj_mult}. For simplicity, we choose
\begin{align}
g(x) &= 
      (1-\gamma) \frac{\E[Y(1)|X=x]}{\E Y(1)} + \gamma \frac{\E[Y(0)|X=x]}{\E Y(0)}  \label{eq:best_g_logs}
\end{align}
where $\E Y(1)$, $\E Y(0)$ are $\mu_1$, $\mu_0$, since for this choice, $\mu_g = 1$.

As an aside, it may seem odd that Equation \eqref{eq:def_Madj_mult} and \eqref{eq:def_Madj_mult_bis} are both asymptotically optimal until we note that $\mu_g = 1$, and hence to first order $\log \mean{G}_1 \approx \mean{G}_1 - 1$ (likewise for $\mean{G}_0$).

\subsection{Multiplicative treatment effect: ratio}
A multiplicative treatment effect can also be posed as a ratio rather than on a log scale. This may be developed as follows:
\begin{align*}
    \delta &= \frac{\mu_1}{\mu_0} - 1 \tag{subtracting 1 so that 0 represents no effect} \\
    f(\mean{Y}_1, \mean{Y}_0) &= \frac{\mean{Y}_1}{\mean{Y}_0} - 1 = \metric\\
    h(\mean{Y}_1, \mean{Y}_0, \mean{G}_1, \mean{G}_0)
    &= \frac{\mean{Y}_1}{\mean{G}_1}/\frac{\mean{Y}_0}{\mean{G}_0} - 1 = \metricadj \\
v_a &= \diffprimea(\mu_1, \mu_0) = \frac{1}{\mu_0} \\
v_b &= \diffprimeb(\mu_1, \mu_0) = -\frac{\mu_1}{\mu_0^2} \\
v_g &= \diffadjprimed(\mu_1, \mu_0, \mu_g, \mu_g) = \frac{\mu_1}{\mu_0 \mu_g}  \\
\frac{\mu_1}{\mu_0 \mu_g}  g(x) 
&= (1-\gamma) \frac{1}{\mu_0} \E[Y(1)|X=x] 
    + \gamma \frac{\mu_1}{\mu_0^2} \E[Y(0)|X=x] \\
\frac{g(x)}{\mu_g} &= 
        (1-\gamma) \frac{\E[Y(1)|X=x]}{\E Y(1)} + \gamma \frac{\E[Y(0)|X=x]}{\E Y(0)}
\end{align*}
where $\E Y(1)$, $\E Y(0)$ are $\mu_1$, $\mu_0$.
This is the same equation that led to the optimal solution for the log scale treatment effects
(Equation \eqref{eq:best_g_logs}). We could have anticipated this by noticing that Theorem \eqref{theorem:optimal_invariance} applies since $\metricadj$ in the two cases are related by a transform.

\subsection{Multiplicative treatment for ratio metrics}
\label{section:ratio}
We now consider multiplicative treatment effects that apply to ratios of sample means.
This is very common in industry where many metrics of interest are ratios. For instance, the average price per item sold or click through rate are examples of ratios of means. In this setting, $Y$ is replaced by a $2$-vector
$\twocolvec{Y}{Z}$. The metric of interest is $\log \E Y - \log \E Z$, and the treatment effect defined as
\begin{align*}
    \delta &= \left(\log \mu_1 - \log \nu_1 \right) - \left(\log \mu_0 - \log \nu_0 \right)
\end{align*}
\begin{align*}
    \textrm{where\ }\twocolvec{\mu_1}{\nu_1} &= \E \twocolvec{Y(1)}{Z(1)} \\
    \textrm{and\ }  \twocolvec{\mu_0}{\nu_0} &= \E \twocolvec{Y(0)}{Z(0)}
\end{align*}

Applying a simple additive adjustment as in Equation \eqref{eq:simple_h}:
\begin{align}
\metric &= f(\twocolvec{\mean{Y}_1}{\mean{Z}_1}, \twocolvec{\mean{Y}_0}{\mean{Z}_0}) 
    = (\log \mean{Y}_1 - \log \mean{Z}_1) - (\log \mean{Y}_0 - \log \mean{Z}_0)  \notag\\
\metricadj &= h(\twocolvec{\mean{Y}_1}{\mean{Z}_1}, \twocolvec{\mean{Y}_0}{\mean{Z}_0}, \mean{G}_1, \mean{G}_0 ) 
   = (\log \mean{Y}_1 - \log \mean{Z}_1 - \mean{G}_1)
   - (\log \mean{Y}_0 - \log \mean{Z}_0 - \mean{G}_0) \label{eq:ratio_h}\\
v_a &= \diffprimea(\twocolvec{\mu_1}{\nu_1}, \twocolvec{\mu_0}{\nu_0}) 
    = \tworowvec{\frac{1}{\mu_1}}{-\frac{1}{\nu_1}} \notag \\
v_b &= \diffprimeb(\twocolvec{\mu_1}{\nu_1}, \twocolvec{\mu_0}{\nu_0})
    = \tworowvec{-\frac{1}{\mu_0}}{\frac{1}{\nu_0}} \notag 
\end{align}
From Equation \eqref{eq:simple_optim} we have
\begin{align}
    g(x)
    &= (1-\gamma) \E \left[\frac{Y(1)}{\E Y(1)} - \frac{Z(1)}{\E Z(1)} \bigg| X=x \right] 
    +  \gamma \E \left[\frac{Y(0)}{\E Y(0)} - \frac{Z(0)}{\E Z(0)} \bigg| X=x \right] 
    \label{eq:ratio_ratio}
\end{align}
We note in passing that Theorem \eqref{theorem:optimal_invariance} extends the result to treatment effects defined as ratios of ratios rather than as log differences:
\begin{align*}
    \metric &= \frac{\mean{Y}_1/\mean{Z}_1}{\mean{Y}_0/\mean{Z}_0} \\
    \metricadj &= \frac{\mean{Y}_1/\mean{Z}_1/\mathrm{e}^{\mean{G}_1}}
                    {\mean{Y}_0/\mean{Z}_0/\mathrm{e}^{\mean{G}_0}}
\end{align*}

It is worth reflecting on the more general implications of what this example shows. Each observation $\twocolvec{y_i}{z_i}$ is a $2$-vector. 
For optimal variance reduction, we may choose a scalar function $g(x)$. If we are to fit $g(x)$ from data, Equation \eqref{eq:ratio_ratio} tells us what the prediction target of that regression should be. Indeed, regardless of the dimensionality of the response, we can achieve asymptotically optimal variance reduction from a scalar prediction.
This is a big simplification, and generalizes to other complex metrics involving multiple responses.

Finally, consider a failure case, one in which Equation \eqref{eq:constraint_optim_g} is not satisfied.
Let
\begin{align*}
f(\twocolvec{\mean{Y}_1}{\mean{Z}_1}, \twocolvec{\mean{Y}_0}{\mean{Z}_0}) 
    &= (\log \mean{Y}_1 - \log \mean{Z}_1) - (\log \mean{Y}_0 - \log \mean{Z}_0) = \metric \\
h(\twocolvec{\mean{Y}_1}{\mean{Z}_1}, \twocolvec{\mean{Y}_0}{\mean{Z}_0}, \mean{G}_1, \mean{G}_0 ) 
    &= (\log \mean{Y}_1 - \log \mean{Z}_1 - \log \mean{G}_1)  - (\log \mean{Y}_0 - \log \mean{Z}_0 - \log \mean{G}_0) 
    = \metricadj \\
v_a &= \diffprimea(\twocolvec{\mu_1}{\nu_1}, \twocolvec{\mu_0}{\nu_0}) 
    = \tworowvec{\frac{1}{\mu_1}}{-\frac{1}{\nu_1}} \\
v_b &= \diffprimeb(\twocolvec{\mu_1}{\nu_1}, \twocolvec{\mu_0}{\nu_0})
    = \tworowvec{-\frac{1}{\mu_0}}{\frac{1}{\nu_0}} \\
v_g &= \diffadjprimed(\twocolvec{\mu_1}{\nu_1}, \twocolvec{\mu_0}{\nu_0}, \mu_g, \mu_g) = \frac{1}{\mu_g}
\end{align*}
The left hand side of Equation \eqref{eq:constraint_optim_g} is $\frac{\mu_g}{\mu_g}=1$,
whereas the right side is
\begin{align*}
 &(1-\gamma) \tworowvec{\frac{1}{\mu_1}}{-\frac{1}{\nu_1}} \twocolvec{\mu_1}{\nu_1}
    -        \gamma \tworowvec{-\frac{1}{\mu_0}}{\frac{1}{\nu_0}} \twocolvec{\mu_0}{\nu_0} \\
    &= (1-\gamma)(1 - 1) - \gamma (-1 + 1) = 0
\end{align*}
Thus, no choice of $g$ can satisfy Equation \eqref{eq:constraint_optim_g} in this case.

\subsection{Algorithm: multiplicative treatment for ratio metrics}
Given the practical importance of ratio metrics, we present an explicit algorithm for this setting that implements Equation \eqref{eq:ratio_ratio}.
This algorithm is largely the same as in Section \eqref{section:additive_algorthm} with adjustments made for both $Y$ and $Z$.
\begin{enumerate}
    \item Divide the data randomly into $K$ equal folds. Let $(T_k, C_k)$ represent the samples in treatment and control for the $k^\mathrm{th}$ fold.
    
    \item Set the training data to all data except the $k^\mathrm{th}$ fold, i.e., $(T-T_k, C-C_k)$.
    
    \item Define 
    \begin{align*}
        {v}(0) &= \frac{y(0)}{\mean{y}(0)} - \frac{z(0)}{\mean{z}(0)} \\
        {v}(1) &= \frac{y(1)}{\mean{y}(1)} - \frac{z(1)}{\mean{z}(1)} 
    \end{align*}
 where $\mean{y}(0)$, $\mean{y}(1)$ are the sample averages of $y$ in control and treatment respectively of the training data. Likewise of $\mean{z}(0)$, $\mean{z}(1)$.
    
    \item Fit a model of the form $\Hat{v}(W,x)$ that minimizes $(\Hat{v}-v)^2$ on the training data,
    where $W=1$ for treatment $W=0$ for control.
    There are no constraints on the structure of this model in terms of how it is trains, regularizes or borrows strength.
    
    \item Using this model, compute
    \begin{align}
    g(x) &=  (1-\gamma) \Hat{v}(W=1,x) +  
    \gamma \Hat{v}(W=0,x)\label{eq:vgnorm}
    \end{align}
    for each observation in $T_k \cup C_k$.
    
    \item Estimate the adjusted metric as per Equation \eqref{eq:ratio_h}:
    \begin{align*}
        \metricadj^k &= \left( \log \frac{1}{|T_k|}\sum_{i \in T_k}y_i  - 
                             \log \frac{1}{|T_k|}\sum_{i \in T_k}z_i -
                             \frac{1}{|T_k|}\sum_{i \in T_k}g(x_i)
                            \right) \\
               &\quad - 
                       \left( \log \frac{1}{|C_k|}\sum_{j \in C_k}y_j  - 
                            \log \frac{1}{|C_k|}\sum_{j \in C_k}z_j -
                                 \frac{1}{|C_k|}\sum_{j \in C_k}g(x_j)
                            \right)
    \end{align*}

    \item Repeat for each fold and form the average $\mean{\metricadj}$. This is the final estimate.
\end{enumerate}
When the data is big (but effects small), $K=2$ will suffice.

\subsection{Additive treatment for ratio metrics}
\label{section:ratio_diff}
One might also imagine an additive treatment effect on ratios, i.e., a difference of ratios; for example, see \cite{ying}. We can apply our framework here as well.

As in Section \ref{section:ratio}, $Y$ is replaced by a $2$-vector $\twocolvec{Y}{Z}$.
The metric of interest is now $\frac{\E Y}{\E Z}$, and the treatment effect defined as
\begin{align*}
    \delta &= \frac{\mu_1}{\nu_1} - \frac{\mu_0}{\nu_0} \\
    \textrm{where\ }\twocolvec{\mu_1}{\nu_1} &= \E \twocolvec{Y(1)}{Z(1)} \\
    \textrm{and\ }  \twocolvec{\mu_0}{\nu_0} &= \E \twocolvec{Y(0)}{Z(0)}
\end{align*}
As before:
\begin{align*}
f(\twocolvec{\mean{Y}_1}{\mean{Z}_1}, \twocolvec{\mean{Y}_0}{\mean{Z}_0}) 
    &= \frac{\mean{Y}_1}{\mean{Z}_1} - \frac{\mean{Y}_0}{\mean{Z}_0} = \metric \\
h(\twocolvec{\mean{Y}_1}{\mean{Z}_1}, \twocolvec{\mean{Y}_0}{\mean{Z}_0}, \mean{G}_1, \mean{G}_0 ) 
    &= \left(\frac{\mean{Y}_1}{\mean{Z}_1} - \mean{G}_1 \right)  
     - \left(\frac{\mean{Y}_0}{\mean{Z}_0} - \mean{G}_0 \right) = \metricadj \\
v_a &= \diffprimea(\twocolvec{\mu_1}{\nu_1}, \twocolvec{\mu_0}{\nu_0}) 
    = \tworowvec{\frac{1}{\nu_1}}{-\frac{\mu_1}{\nu_1^2}} \\
v_b &= \diffprimeb(\twocolvec{\mu_1}{\nu_1}, \twocolvec{\mu_0}{\nu_0})
    = \tworowvec{-\frac{1}{\nu_0}}{\frac{\mu_0}{\nu_0^2}} \\
v_g &= \diffadjprimed(\twocolvec{\mu_1}{\nu_1}, \twocolvec{\mu_0}{\nu_0}, \mu_g, \mu_g) = 1
\end{align*}
\begin{align*}
g(x)
    &=  (1-\gamma) \E \left[\frac{Y(1)}{\nu_1} - \frac{\mu_1}{\nu_1^2}Z(1) \bigg| X=x \right]
    +       \gamma \E \left[\frac{Y(0)}{\nu_0} - \frac{\mu_0}{\nu_0^2}Z(0) \bigg| X=x \right]
\end{align*}
\begin{align*}
g(x)
    =  (1-\gamma)&\left( \frac{\E Y(1)}{\E Z(1)} \right)\E \left[\frac{Y(1)}{\E Y(1)} - \frac{Z(1)}{\E Z(1)} \bigg| X=x \right] \\
    +  \gamma &\left(\frac{\E Y(0)}{\E Z(0)} \right) \E \left[\frac{Y(0)}{\E Y(0)} - \frac{Z(0)}{\E Z(0)} \bigg| X=x \right]
\end{align*}

\subsection{Stable denominators and the like}
Some metrics involve random variables that are not affected by the treatment but nonetheless jointly vary with the response. For instance, we may wish to measure how a change in an app affects the days of user activity $Y$. However, if users come into the experiment at different times, we may wish to normalize by the number of days since first exposure (of treatment or control) $Z$ where $Z \geq 1$. This is sometimes referred to as the ``stable denominator assumption'' (e.g. \cite{ying} use it in a restricted sense). We use this example to illustrate how covariates not affected by the experiment may nonetheless play a role in the definition of treatment effect.

Two possible ways to define the metric ``Fraction of active days since exposure'' are $\frac{\E Y}{\E Z}$ and $\E \left(\frac{Y}{Z}\right)$. In the first case, the (multiplicative) treatment effect would be
\begin{align*}
\delta &= \log \frac{\E Y(1)}{\E Z} - \log \frac{\E Y(0)}{\E Z} \\
&= \log \E Y(1) - \log \E Y(0)
\end{align*}
Here $Z(1)$ and $Z(0)$ are replaced by $Z$ because it is unaffected by treatment. 
This case reduces to the simple (i.e., non-ratio) case of multiplicative treatment effects in Section \ref{section:mult_simple}. The only difference is that we may now use $Z$ as a covariate (i.e, a component of X), since it is unaffected by treatment. If $Y$ tends to be proportional to $Z$, we can leverage this knowledge directly by modeling
$Y \sim \beta Z$. But we are not constrained to do so. For example, if the relationship between $Y$ and $Z$ is sublinear for larger $Z$
might use
$Y \sim \beta Z^\eta$ for $0 < \eta < 1$.
The larger point here is that we do not have to embed data intuition in the metric itself and can let data decide.

The second definition of ``Fraction of active days since exposure'' leads to
\begin{align*}
\delta &= \log \E \left(\frac{Y(1)}{Z}\right) - \log \E \left(\frac{Y(0)}{Z}\right) \\
&= \log \E V(1) - \log \E V(0) \tag{where $V = \frac{Y}{Z}$}
\end{align*}
Here, $Z$ is both absorbed into the response $V$ and may also used to predict $V$.

\subsection{Ratio of variances}
The experimental effects of interest in most applications involve  means, either the difference or ratio of average response between treatment and control. Nevertheless, it is possible to use variance reduction to estimate the change in variance of a response. Even if the difference in group means is not significant, there may be an underlying difference in distribution that is most readily detected by a difference in variance of the response.

For instance, suppose a particular change to a product meant to make things easier for less experienced users works as intended for less experienced users, but negatively affects more experienced users. And further suppose that these effects cancel out in a metric that measures user benefit. Such an experiment may have little average treatment effect, but it might still have a sizable impact on different user segments. A statistically significant effect on each of several user segments would directly demonstrate this heterogeneous treatment effect, but may be hard to achieve due to sample size. On the other hand, a statistically significant change (in our example a reduction) in variance may be easier observed.

Consider the case in which the log difference of variances is of interest. We can write the estimand as
\begin{align*}
\delta &= \log\left(\E (Y(1)^2) - (\E Y(1))^2\right) 
        - \log\left(\E (Y(0)^2) - (\E Y(0))^2\right)
\end{align*}
Our variance reduction framework works for estimators that use sample means. We can express a simple estimator in this manner by defining the response as the tuple $\twocolvec{y}{y^2}$ whose expectation gives us $\twocolvec{\E Y}{\E Y^2}$, terms required for variance.
Let the sample means in treatment and control be $\twocolvec{\mean{Y}_1}{\mean{Y_1^2}}$ and $\twocolvec{\mean{Y_0}}{\mean{Y_0^2}}$ respectively. Thus
\begin{align*}
\metric &= f(\twocolvec{\mean{Y}_1}{\mean{Y_1^2}}, \twocolvec{\mean{Y}_0}{\mean{Y_0^2}})
 = \log \left( \mean{Y_1^2} - (\mean{Y}_1)^2 \right) - \log \left( \mean{Y_0^2} - (\mean{Y}_0)^2 \right)\\
\metricadj &= h(\twocolvec{\mean{Y}_1}{\mean{Y_1^2}}, \twocolvec{\mean{Y}_0}{\mean{Y_0^2}}, \mean{G_1}, \mean{G_0})
 = \log \left( \mean{Y_1^2} - (\mean{Y}_1)^2 \right) - \log \left( \mean{Y_0^2} - (\mean{Y}_0)^2 \right) - \mean{G_1} + \mean{G_0}
\end{align*}
Let the mean and variance of $Y(1)$ be $\mu_1$, $\sigma_1^2$ and for $Y(0)$ be $\mu_0$, $\sigma_0^2$. Then following through we find
\begin{align*}
    v_a &= \frac{1}{\sigma_1^2} \tworowvec{-2\mu_1}{1} \\
    v_b &= \frac{1}{\sigma_0^2} \tworowvec{2\mu_0}{-1} \\
    v_g &= 1 \\
    g(x) &= (1-\gamma) \frac{1}{\sigma_1^2}\E[Y(1)^2 - 2\mu_1 Y(1) |X=x] 
              + \gamma \frac{1}{\sigma_0^2}\E[Y(0)^2 - 2\mu_0 Y(0) |X=x] 
\end{align*}
Since our $h$ is invariant to an additive constant in $g(x)$ we can simplify:
\begin{align*}
    g(x) &= (1-\gamma) \E \left[\left(\frac{Y(1) - \mu_1}{\sigma_1}\right)^2 \bigg|X=x \right]
              + \gamma \E \left[\left(\frac{Y(0) - \mu_0}{\sigma_0}\right)^2 \bigg|X=x \right]
\end{align*}
Of course, $\mu_1$, $\sigma_1$, $\mu_0$, $\sigma_0$ are unknown. We use cross-fitting as in \eqref{section:additive_algorthm} to estimate these parameters and choose $g$ with one portion of the data and apply it to a different portion of the data to make a variance-reduced estimate.

\subsection{Regression and trend estimation}
In certain situations, the experiment effect we wish to estimate is the coefficient of a regression model. One such case is an estimate of the time trend in daily comparisons between treatment and control. For instance, it might be the case that a new feature takes time to affect the response. The trend may become obvious only after a considerable amount of time, and even then, eye-balling the trend in a noisy time series is challenging. It is preferable to make a statistical estimate of the trend directly, along with a measure of statistical uncertainty.

Imagine that a treatment effect has a value at different points in time such that $\delta_t$ denotes the treatment effect at discrete time unit $t$. Assume $t \in \mathbb{Z}$ indexes the day of the experiment. A linear trend is a linear relationship between $t$ and $\delta_t$, i.e., $\beta_1$ estimated by fitting the model
\begin{align*}
    \delta_t &= \beta_0 + \beta_1 t
\end{align*}
Likewise for quadratic or other function of $t$. We can estimate trend coefficients by regressing the daily estimates of treatment effect $\metric_t$ against functions of $t$. In the linear case, the coefficients $\beta_0$ and $\beta_1$ are estimated by regression against a design matrix $A$ whose columns are $1$ and the linear sequence of $t$. For example, below is $A$ for an experiment lasting five days:
\begin{align*}
A &= 
\begin{bmatrix}
1 & 1 \\ 
1 & 2 \\
1 & 3 \\
1 & 4 \\
1 & 5
\end{bmatrix}
\end{align*}
The regression coefficients $\hat{\beta}$ are given by the well-known least squares estimate
\begin{align*}
    \hat{\beta} &= (A^T A)^{-1} A^T \metric \\
    &= \begin{bmatrix}
    0.8 & 0.5 & 0.2 & -0.1 & -0.4 \\
   -0.2 & -0.1 & 0.0 & 0.1 & 0.2
    \end{bmatrix} \metric \tag{for 5 time points}
\end{align*}
where  $\metric$ is a column vector of daily estimates of the treatment effect. The point is that $(A^T A)^{-1} A^T$ is just a constant, depending only on the number of time points. Thus, the regression estimate of the linear trend $\beta_1$ is a fixed linear combination of the daily estimates of the experiment effect:
\begin{align*}
    \hat{\beta}_1 &= a^T \metric
\end{align*}
Likewise for any other regression coefficients.

Suppose we wish to use the framework of this paper to reduce the variance of the estimate of $\beta_1$. One way is simply to apply variance reduction independently on each day, and then estimate the linear trend $\beta_1$ as a linear combination of these adjusted daily estimates. If the daily estimates are unbiased, the estimate of trend will also be unbiased. In this approach, we would choose a $g$ for each day independently, call it $g_t$. Perhaps this works out fine, in which case there is nothing more to say. But given that the modeling for each $g_t$ uses only data from a single day, the model may suffer from data sparseness. The question is whether there is another approach that allows us to borrow strength across days. When $h$ is an additive adjustment as in Equation \eqref{eq:simple_h}, there is indeed a simpler solution.

We start with a concrete example, say, a multiplicative treatment effect on a log scale. Earlier, we found that the optimal $g$ is given by Equation \eqref{eq:best_g_logs}.
Now we have daily estimates such that we index all variables with an additional $t$. Thus
\begin{align*}
    \metric_t &= \log \mean{Y}_{t,1} - \log \mean{Y}_{t,0} \\
    \metric_{\mathrm{adj},t} &= \log \mean{Y}_{t,1} - \log \mean{Y}_{t,0}  - \mean{G}_{t,1} + \mean{G}_{t,0}
\end{align*}
where $\mean{Y}_{t,1}$, $\mean{Y}_{t,0}$ are the average response on Day $t$ in treatment, control and  $\mean{G}_{t,1}$, $\mean{G}_{t,0}$ are the average value of $g_t(x)$ evaluated on treatment, control.

Let $\metric$ and $\metricadj$ denote the column vector of daily values $\metric_t$ and $\metric_{\mathrm{adj},t}$ respectively. Then the adjusted estimate for the linear trend coefficient is
\begin{align*}
    \mathrm{Adjusted\ }\hat{\beta}_1 &= a^T \metricadj \\
    &= \sum_t a_t \metric_{\mathrm{adj},t} \tag{ $a_t$ is the $t^\textrm{th}$ daily component of $a$}\\
    &= \sum_t a_t \left(\log \mean{Y}_{t,1} - \log \mean{Y}_{t,0}  - \mean{G}_{t,1} + \mean{G}_{t,0} \right) \\
    &= a^T \metric - \sum_t a_t \mean{G}_{t,1} + \sum_t a_t  \mean{G}_{t,0} \\
    &= a^T \metric - \mean{G}_1 + \mean{G}_0
\end{align*}
where $\mean{G}_1$, $\mean{G}_0$ are $g(x)$ evaluated over treatment, control and
\begin{align*}
    g(x) &= \sum_t a_t g_t(x)
\end{align*}
In other words, linearity allows us to define a single optimal $g$ as the weighted sum of daily $g_t$.

Recall that for multiplicative treatment effect on a log scale, Equation \eqref{eq:best_g_logs} gives us single-day optimal $g$ as follows
\begin{align}
g_t(x) &= 
      (1-\gamma) \frac{\E[Y_t(1)|X=x]}{\E Y_t(1)} + \gamma \frac{\E[Y_t(0)|X=x]}{\E Y_t(0)} \notag \\
 g(x) &= \sum_t a_t \left(
      (1-\gamma) \frac{\E[Y_t(1)|X=x]}{\E Y_t(1)} + \gamma \frac{\E[Y_t(0)|X=x]}{\E Y_t(0)}
      \right) \notag \\
g(x) &=
      (1-\gamma) \E \left[ \sum_t a_t \frac{Y_t(1)}{\E Y_t(1)} \bigg| X = x \right]
        + \gamma \E \left[ \sum_t a_t \frac{Y_t(0)}{\E Y_t(0)} \bigg| X = x \right] \label{eq:g_trend}
\end{align}
where $Y_t(1)$, $Y_t(0)$ are the individual random responses in treatment, control on Day $t$. Thus Equation \eqref{eq:g_trend} gives us a single optimal $g$ for the entire experiment.

To model $g(x)$ we first define
\begin{align*}
v(W, x) &= \E \left[ \sum_t a_t \frac{Y_t(W)}{\E Y_t(W)} \bigg| X = x \right]
\end{align*}
where $W=1$ for treatment and $W=0$ for control. Again using cross-fitting, we estimate $\hat{v}(W, x)$ by means of a prediction model that minimizes mean squared error. The prediction target for each observation is a weighted sum of daily responses scaled by average daily response. Finally, we estimate
\begin{align*}
    g(x) &=  (1-\gamma)\Hat{v}(W=1,x) +  \gamma \Hat{v}(W=0,x)
\end{align*}

\section{Practical considerations}
\subsection{Modeling}
While there are no constraints on the model $\Hat{y}(W,x)$, there are modeling choices to be made. For instance, to what extent we should borrow strength across treatment and control observations depends on how much data there is in the smaller arm versus how different the data is in the two arms.

We therefore recommend a model that shrinks predictions towards a pooled model (e.g., through regularization), especially if treatment effects tend to be small. Given how it is used, such a model will only affect efficiency and not cause bias in our estimate of the treatment effect. Furthermore, any loss in efficiency for large effects may be an acceptable trade off, large effects being inherently easier to identify.
Some additional choices are worth discussing:
\begin{itemize}
    \item One complication to modeling is that the relative weights of the two arms of the experiment move in the ``wrong" (undesirable) direction. That is, the smaller $\gamma$, the less of it we have as a fraction of the total data but the greater its contribution to $g(x)$ as shown in Theorem \eqref{theorem:General}. We will call this phenomenon the \textbf{wrong direction of weights}.
    
    \item Since treatment and control contribute unequally to $g(x)$ when $\gamma\neq \frac{1}{2}$, it is reasonable to give weights in ratio $(1-\gamma):\gamma$ to observations in treatment\,:\,control for the model $\Hat{y}(W, x)$ (or $\Hat{v}(W, x)$ for ratios), though this also an empirical question. When $\gamma$ is far from $\frac{1}{2}$, the model should borrow strength across arms.
    
    \item The logical extreme of regularizing towards a pooled model is a simple pooled model itself, in which $W$ plays no role. Whether this simplification is worthwhile is an empirical question. But in general, it will not provide maximum variance reduction if the effects are large. Furthermore, if the arms of the experiment are unequal, the wrong direction of weights may exacerbate the situation.
    
    \item Given the unequal contribution of treatment and control to $g(x)$, it is tempting to try to model $\Hat{g}$ directly by building a pooled model with relative weights for treatment\,:\,control of $(1-\gamma)^2:\gamma^2$. The problem with a model trained on this data is that it will borrow vanishing strength from the larger control arm as $\gamma \to 0$. Vice versa when $\gamma \to 1$.

    \item More complex prediction models require more data, thereby warranting $K > 2$. But given diminishing returns, there is not much reason for $K$ to be very large.
    
    \item Beyond variance reduction, an explicit model of treatment effects may also help identify where the treatment effect is strongest. This gives additional guidance to the experimenter.
\end{itemize}

\subsection{Prediction accuracy and variance reduction}
Thus far, we have sought $g$ that delivers optimal variance reduction. But for practical applications, we would like to the know the impact on variance reduction if we use $\widehat{g}$ that deviates from optimality. From Equation \eqref{eq:g_suboptim}
\begin{align*}
\VarDM \seq{\metricadj} &= \VarDM \seq{\metricadj^*} + 
\left(\frac{1}{\lambda_1} + \frac{1}{\lambda_0} \right)
\Var \left(v_g (\widehat{g}(X) - g(X)) \right) 
\end{align*}
Assuming $\widehat{g}(X)$ is unbiased for $g(X)$, the additional variance incurred for a suboptimal $g$ is
\begin{align*}
v_g^2\left(\frac{1}{\lambda_1} + \frac{1}{\lambda_0} \right)   
\E(\widehat{g}(X) - g(X))^2
\end{align*}
which is directly proportional to the mean squared error (MSE) in $\widehat{g}(X)$.

\subsection{Variance reduction potential}
Variance reduction in large-scale online systems requires an investment in engineering infrastructure. Engineering costs will depend on the set of covariates employed. It is therefore important to size the opportunity for variance reduction prior to any investment. This can be achieved by building a model using the proposed set of covariates to predict the response.

If data from past experiments is available, it may be possible to estimate $\VarDM \seq{\metricadj}$ directly from Equation \eqref{eq:g_suboptim}. If the data for comparable experiments is limited, we can estimate the impact of variance reduction under the assumption that the response under treatment $Y(1)$ does not have a significantly different distribution than $Y(0)$ beyond a small additive effect. The logic here is that variance reduction is most necessary when the treatment effects are small, and hence don't change the distribution of $Y$.

Let's look at this assumption in two different cases: first the case of additive treatment effects.
From Section \ref{section:addtive}, we have $v_a = 1$, $v_b=-1$, $v_g=1$:
\begin{align*}
    \VarDM \seq{\metricadj}
    &= 
        \frac{1}{\lambda_1} \Var \left(Y(1) - \widehat{g}(X) \right) + 
        \frac{1}{\lambda_0} \Var \left(Y(0) - \widehat{g}(X) \right) 
        \tag{from Equation \eqref{eq:vardm_optim}} \\
    &\approx \frac{1}{\lambda_1} \Var \left(Y(0) - \widehat{g}(X) \right) + 
        \frac{1}{\lambda_0} \Var \left(Y(0) - \widehat{g}(X) \right) 
        \tag{assuming similar distributions}\\
    &= \left(\frac{1}{\lambda_1} + \frac{1}{\lambda_0}\right)
        \Var \left(Y(0) - \widehat{g}(X) \right)
\end{align*}
Since $g$ is a convex combination of $\E Y(1)$ and $\E Y(0)$ up to an additive constant, $g$ reduces in this case simply to $g(x) = \E [Y(0)|X=x]$,
i.e., to predicting $Y(0)$ from $X$, and that is what we see here.

Meanwhile unadjusted variance is
\begin{align*}
\VarDM \seq{\metric}
    &= 
        \frac{1}{\lambda_1} \Var Y(1) + \frac{1}{\lambda_0} \Var Y(0) \\
    &\approx 
        \left(\frac{1}{\lambda_1} + \frac{1}{\lambda_0}\right) \Var Y(0) \\
\Rightarrow \frac{\VarDM \seq{\metricadj}}{\VarDM \seq{\metric}}
&\approx \frac{\Var \left(Y(0) - \widehat{g}(X) \right)}{\Var Y(0)} \\
&=  \frac{\E \left(Y(0) - \widehat{g}(X) \right)^2}{\Var Y(0)}
\end{align*}
where the last step assumes $\widehat{g}$ is unbiased for $g$.
In other words, the proportion $\VarDM \seq{\metric}$ remaining is approximately equal to the MSE in predicting $Y$ as a fraction of its original variance. This gives us an estimate for how much variance reduction we can expect. For difference effects, the asymptotic variance matches the small sample variance (i.e., $\VarDM = \Var$) and this result can be derived more directly.

We can repeat this exercise for multiplicative treatment effect for ratios. Following the treatment in Section \ref{section:ratio}, we arrive at a similar result:
\begin{align*}
\frac{\VarDM \seq{\metricadj}}{\VarDM \seq{\metric}}
&\approx \frac{\E \left(V(0) - \widehat{v}(X) \right)^2}{\Var V(0)} \\
\mathrm{where\ }V(0) &= \frac{Y(0)}{\E Y(0)} - \frac{Z(0)}{\E Z(0)}
\end{align*}
In this case we must rely on asymptotic variance.

\section{Comparison with other methods}
Earlier work on variance reduction in online experiments involved simple models in which there is a linear relationship between scalar $X$ and $Y$. The CUPED algorithm \citep{cuped} is one example, that uses a common model for $Y(0)$ and $Y(1)$. A more general model is given by \citep{soriano}, in which there is a different coefficient for $Y(0)$ and $Y(1)$. Due to the scalar nature of $X$ considered, the models are employed with $X$ as pre-experiment values of the metric represented by $Y$. But one advantage of such simple models is that they do not need to concern themselves with bias due to over-fitting or regularization. Furthermore, they have computational advantage in being able to work with aggregates (sums being sufficient statistics for linear regression). See also \cite{lin2013}, which similarly minimizes asymptotic variance with linear model assumptions.

When computation is not a concern and large models are employed, model bias is a central problem to be addressed. One theoretic approach to bias is given in \citep{hosseini}. Other approaches are inspired by Double ML \citep{doubleml}. Even though Double ML was designed for causal inference without randomization, we may apply it directly to the problem of variance reduction as follows:
\begin{enumerate}
    \item Regress $Y$ on $X$ using cross-fitting to estimate residuals $\Tilde{Y} = Y - \Hat{Y}(X)$. The model pools both arms of the experiment without distinction.
    \item Compute residual for treatment variable $W$ as $\Tilde{W} = W - \E W$.
    \item Regress $\Tilde{Y}$ on $\Tilde{W}$. The regression coefficient is an estimate of the additive treatment effect.
\end{enumerate}
Such an approach works because the cross-fitting removes bias due to over-fitting and regularization subject to some constraints on the models used (see \citep{doubleml} for details). Modulo model bias elimination, the approach is equivalent to building a pooled model for the portion of $Y$ predictable by $X$. As such, it has the limitation of the pooled model when $p$ is far from $\frac{1}{2}$.

\citep{ying} uses an approach inspired by Double ML, and comes closest to this paper. Their idea is to exploit the semi-parametric optimality of the AIPW estimator \citep{robins}, \citep{davidian}. The AIPW estimator is typically employed to estimate causal effects from observational data. By using both propensity scores and outcome models, AIPW is doubly robust to errors in either one. The AIPW estimator is given by  Equation (9) of \citep{davidian} can be written in our notation as
\begin{align*}
   \mathit{AIPW} = \frac{1}{|T|+|C|}  \sum_{i \in T \cup C} 
   &\frac{W_i Y_i - (W_i - \Hat{p}(X_i)) \Hat{y}(1, X_i)}{\Hat{p}(X_i)} \\
   &-\frac{(1-W_i) Y_i + (W_i - \Hat{p}(X_i)) \Hat{y}(0, X_i)}{1-\Hat{p}(X_i)} 
\end{align*}
where $W_i$ is the assignment of the experiment sample $i$,
and $\Hat{p}(x)$, $\Hat{y}(w, x)$ are  estimated propensity score and estimated outcome model.
When $\Hat{p}(x) = p$ for a randomized experiment, this expression reduces to our Equation \eqref{eq:g_diff}.
Following this approach for additive treatment effect, \citep{ying} build two models $\Hat{y}(0)(x)$ and $\Hat{y}(1)(x)$ using cross-fitting. This approach does not generalize to multiplicative effects.

Except for \citep{soriano}, none of the works cited deal with multiplicative treatment effects. Asymptotic optimality demonstrated in this paper is a primary contribution. Only \citep{ying} deals with large models for ratio metrics (with additive treatment effects) but their treatment requires training models for both $Y$ and $Z$ (for each of treatment and control). In contrast, our solution is simpler, as are our proofs. Moreover, prior work either models a single arm (e.g. \citep{hosseini}) or requires models for treatment and control to be trained separately (e.g. \citep{ying}). The former is suboptimal while the latter is problematic when one arm is much smaller than the other. Left unanswered is the practical modeling question of how best to borrow strength across models for $\Hat{y}(W=1,x)$ and $\Hat{y}(W=0,x)$. By focusing our attention specifically on the properties of optimal $g(x)$, we are able to shed light on this area.

\bibliographystyle{apalike} 
\bibliography{References}

@book{imbens_rubin, place={Cambridge}, title={Causal Inference for Statistics, Social, and Biomedical Sciences: An Introduction}, DOI={10.1017/CBO9781139025751}, publisher={Cambridge University Press}, author={Imbens, Guido W. and Rubin, Donald B.}, year={2015}}

@article{doubleml,
    author = {Chernozhukov, Victor and Chetverikov, Denis and Demirer, Mert and Duflo, Esther and Hansen, Christian and Newey, Whitney and Robins, James},
    title = "{Double/debiased machine learning for treatment and structural parameters}",
    journal = {The Econometrics Journal},
    volume = {21},
    number = {1},
    pages = {C1-C68},
    year = {2018},
    month = {01},
    issn = {1368-4221},
    doi = {10.1111/ectj.12097},
    url = {https://doi.org/10.1111/ectj.12097},
    eprint = {https://academic.oup.com/ectj/article-pdf/21/1/C1/27684918/ectj00c1.pdf},
}

@inproceedings{cuped,
  author = {Deng, Alex and Xu, Ya and Kohavi, Ron and Walker, Toby},
  title = {Improving the sensitivity of online controlled experiments by utilizing pre-experiment data},
  year = {2013},
  booktitle = {Proceedings of the Sixth ACM International Conference on Web Search and Data Mining},
  doi = {10.1145/2433396.2433413}
}

@article{ying,
author = {Ying Jin and Shan Ba},
title = {Toward Optimal Variance Reduction in Online Controlled Experiments},
journal = {Technometrics},
volume = {65},
number = {2},
pages = {231--242},
year = {2023},
publisher = {Taylor \& Francis},
doi = {10.1080/00401706.2022.2142670},
URL = {https://doi.org/10.1080/00401706.2022.2142670},
eprint = {https://doi.org/10.1080/00401706.2022.2142670}
}

@article{lin2013,
author = {Winston Lin},
title = {{Agnostic notes on regression adjustments to experimental data: Reexamining Freedman’s critique}},
volume = {7},
journal = {The Annals of Applied Statistics},
number = {1},
publisher = {Institute of Mathematical Statistics},
pages = {295 -- 318},
year = {2013},
doi = {10.1214/12-AOAS583},
URL = {https://doi.org/10.1214/12-AOAS583}
}

@misc{hosseini,
Author = {Reza Hosseini and Amir Najmi},
Title = {Unbiased variance reduction in randomized experiments},
Year = {2019},
Eprint = {arXiv:1904.03817},
howpublished = "\url{https://arxiv.org/abs/1904.03817}"
}

@misc{soriano,
Author = {Jacopo Soriano},
Title = {Percent Change Estimation in Large Scale Online Experiments},
Year = {2017},
Eprint = {arXiv:1711.00562},
howpublished = "\url{https://arxiv.org/abs/1711.00562}"
}

@article{davidian,
author = {Lunceford, Jared K. and Davidian, Marie},
title = {Stratification and weighting via the propensity score in estimation of causal treatment effects: a comparative study},
journal = {Statistics in Medicine},
volume = {23},
number = {19},
pages = {2937-2960},
doi = {https://doi.org/10.1002/sim.1903},
url = {https://onlinelibrary.wiley.com/doi/abs/10.1002/sim.1903},
eprint = {https://onlinelibrary.wiley.com/doi/pdf/10.1002/sim.1903},
year = {2004}
}

@article{robins,
author = { James M.   Robins  and  Andrea   Rotnitzky  and  Lue   Ping   Zhao },
title = {Estimation of Regression Coefficients When Some Regressors are not Always Observed},
journal = {Journal of the American Statistical Association},
volume = {89},
number = {427},
pages = {846-866},
year  = {1994},
publisher = {Taylor & Francis},
doi = {10.1080/01621459.1994.10476818},
URL = { https://doi.org/10.1080/01621459.1994.10476818 },
eprint = { https://doi.org/10.1080/01621459.1994.10476818 }
}

\end{document}